\documentclass[11pt,a4paper]{article}
\pdfoutput=1

\usepackage[vmargin=33mm,hmargin=33mm]{geometry}
\usepackage{amsmath,amssymb,amsthm,booktabs,color,doi,enumitem,graphicx,latexsym,stmaryrd,url,xspace}
\usepackage[numbers,sort&compress]{natbib}
\usepackage[basic,small]{complexity}
\usepackage[small,margin=1cm]{caption}
\usepackage{microtype}
\usepackage{hyperref}

\newclass{\LOCAL}{LOCAL}
\newclass{\ID}{ID}
\newclass{\OI}{OI}
\newclass{\PO}{PO}
\newclass{\EC}{EC}

\newcommand{\FM}{{\small FM}\xspace}

\newcommand{\Alg}{\ensuremath{\mathcal{A}}}

\newcommand{\GG}{G\mspace{-1mu}G}
\newcommand{\HH}{H\mspace{-3mu}H}
\newcommand{\GH}{G\mspace{-2mu}H}
\newcommand{\N}{\ensuremath{\mathbb{N}}}

\DeclareMathOperator{\dist}{dist}

\newtheorem{theorem}{Theorem}
\newtheorem{corollary}[theorem]{Corollary}
\newtheorem{lemma}[theorem]{Lemma}
\newtheorem{fact}[theorem]{Fact}

\newtheorem{step}{Step}
\newtheorem*{question}{Open question}

\theoremstyle{definition}
\newtheorem{definition}{Definition}

\allowdisplaybreaks

\hypersetup{
    colorlinks=true,
    linkcolor=black,
    citecolor=black,
    filecolor=black,
    urlcolor=[rgb]{0,0.1,0.5},
    pdfauthor={Mika G{\"o}{\"o}s, Juho Hirvonen, Jukka Suomela},
    pdftitle={Linear-in-Delta Lower Bounds in the LOCAL Model},
}

\newenvironment{myabstract}
               {\list{}{\listparindent 1.5em%
                        \itemindent    \listparindent
                        \leftmargin    0pt
                        \rightmargin   0pt
                        \parsep        0pt}%
                \item\relax}
               {\endlist}

\newenvironment{mycover}
               {\list{}{\listparindent 0pt
                        \itemindent    \listparindent
                        \leftmargin    0pt
                        \rightmargin   0pt
                        \parsep        0pt}%
                \raggedright
                \item\relax}
               {\endlist}

\begin{document}

\vspace*{2ex}
\begin{mycover}
{\LARGE \textbf{Linear-in-\texorpdfstring{$\Delta$}{Delta} Lower Bounds in the {\sf LOCAL} Model}}
\bigskip
\bigskip

\textbf{Mika G\"o\"os} \\
{\small Department of Computer Science, University of Toronto, Canada\\
\nolinkurl{mika.goos@mail.utoronto.ca}\par}
\bigskip

\textbf{Juho Hirvonen} \\
{\small Helsinki Institute for Information Technology HIIT,\\Department of Computer Science, University of Helsinki, Finland\\
\nolinkurl{juho.hirvonen@cs.helsinki.fi}\par}
\bigskip

\textbf{Jukka Suomela} \\
{\small Helsinki Institute for Information Technology HIIT,\\Department of Computer Science, University of Helsinki, Finland\\
\nolinkurl{jukka.suomela@cs.helsinki.fi}\par}
\end{mycover}
\bigskip
\begin{myabstract}
\noindent\textbf{Abstract.}
By prior work, there is a distributed algorithm that finds a maximal fractional matching (maximal edge packing) in $O(\Delta)$ rounds, where $\Delta$ is the maximum degree of the graph. We show that this is optimal: there is no distributed algorithm that finds a maximal fractional matching in $o(\Delta)$ rounds.

Our work gives the first linear-in-$\Delta$ lower bound for a natural graph problem in the standard model of distributed computing---prior lower bounds for a wide range of graph problems have been at best logarithmic in $\Delta$.
\end{myabstract}
\thispagestyle{empty}
\setcounter{page}{0}
\newpage

\section{Introduction}

This work settles the distributed time complexity of the maximal fractional matching problem as a function of $\Delta$, the maximum degree of the input graph.

By prior work~\cite{astrand10vc-sc}, there is a distributed algorithm that finds a maximal fractional matching (also known as a maximal edge packing) in $O(\Delta)$ communication rounds, independently of the number of nodes. In this work, we show that this is optimal: there is no distributed algorithm that finds a maximal fractional matching in $o(\Delta)$ rounds.

This is the first linear-in-$\Delta$ lower bound for a natural graph problem in the standard $\LOCAL$ model of distributed computing. It is also a step towards understanding the complexity of the non-fractional analogue, the maximal matching problem, which is a basic symmetry breaking primitive in the field of distributed graph algorithms. For many related primitives, the prior lower bounds in the $\LOCAL$ model have been at best logarithmic in $\Delta$.

\subsection{Matchings}

Simple randomised distributed algorithms that find a maximal matching in time $O(\log n)$ have been known since the 1980s \cite{alon86fast,israeli86matching,luby86simple}. Currently, the fastest algorithms that compute a maximal matching stand as follows:
\begin{itemize}[label=$-$]
\item {\bf Dense graphs.} There is a recent $O(\log\Delta + \log^4\log n)$-time randomised algorithm due to Barenboim et al.~\cite{barenboim12locality}. The fastest known deterministic algorithm runs in time $O(\log^4 n)$ and is due to Ha{\'n}{\'c}kowiak et al.~\cite{hanckowiak01distributed}.
\item {\bf Sparse graphs.} There is a $O(\Delta + \log^*n)$-time deterministic algorithm due to Panconesi and Rizzi~\cite{panconesi01some}. Here $\log^* n$ is the iterated logarithm of $n$, a very slowly growing function. 
\end{itemize}

Our focus is on the sparse case. It is a long-standing open problem to either improve on the $O(\Delta+\log^*n)$-time algorithm of Panconesi and Rizzi, or prove it optimal by finding a matching lower bound. In fact, Linial's~\cite{linial92locality} seminal work already implies that $O(\Delta) + o(\log^* n)$ rounds is not sufficient. This leaves us with the following possibility (see Barenboim and Elkin \cite[Open Problem 10.6]{barenboim13distributed}):
\begin{question}
Can maximal matchings be computed in time $o(\Delta)+O(\log^*n)$?
\end{question}

We conjecture that there are no such algorithms. The lower bound presented in this work builds towards proving conjectures of this form.

\subsection{Fractional matchings}

While a matching associates a weight $0$ or $1$ with each edge of a graph, with $1$ indicating that the edge is in a matching, a fractional matching (\FM) associates a weight between $0$ and $1$ with each edge. In both cases, the total weight of the edges incident to any given node has to be at most~$1$.

Formally, let $G = (V,E)$ be a simple undirected graph and let $y\colon E \to [0,1]$ associate weights to the edges of $G$. Define, for each $v\in V$,
\[
    y[v] := \sum_{e \in E: v \in e} y(e).
\]
The function $y$ is called a \emph{fractional matching}, or an \FM for short, if $y[v] \le 1$ for each node $v$. A node $v$ is \emph{saturated} if $y[v]=1$.

There are two interesting varieties of fractional matchings.
\begin{itemize}[label=$-$]
\item {\bf Maximum weight.}
An \FM $y$ is of \emph{maximum weight}, if its total weight $\sum_{e\in E} y(e)$ is the maximum over all fractional matchings on $G$.
\item {\bf Maximality.}
An \FM $y$ is \emph{maximal}, if each edge $e$ has at least one saturated endpoint $v \in e$.
\end{itemize}
See below for examples of (a)~a maximum-weight \FM{}, and (b)~a maximal \FM{}; the saturated nodes are highlighted.
\begin{center}
    \includegraphics[page=1]{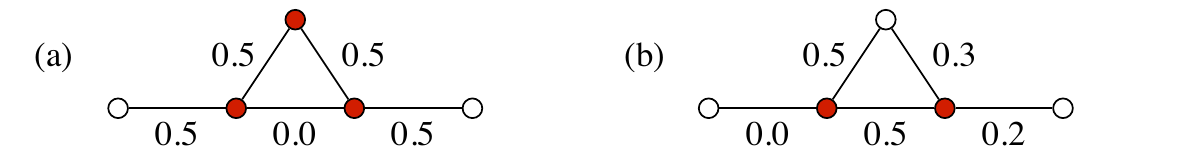}
\end{center}

\paragraph{Distributed complexity.}
The distributed complexity of computing maximum-weight \FM{}s is completely understood. It is easy to see that computing an exact solution requires time $\Omega(n)$ already on odd-length path graphs. If one settles for an approximate solution, then \FM{}s whose total weight is at least a $(1-\epsilon)$-fraction of the maximum can be computed in time $O(\epsilon^{-1}\log\Delta)$ by the well-known results of Kuhn et al.~\cite{kuhn04what,kuhn06price,kuhn10local}. This is optimal: Kuhn et al.\ also show that any constant-factor approximation of maximum-weight \FM{}s requires time $\Omega(\log\Delta)$.

By contrast, the complexity of computing maximal \FM{}s has not been understood. A maximal \FM is a $1/2$-approximation of a maximum-weight \FM, so the results of Kuhn et al.\ imply that finding a maximal \FM requires time $\Omega(\log \Delta)$, but this lower bound is exponentially small in comparison to the $O(\Delta)$ upper bound~\cite{astrand10vc-sc}.

\subsection{Contributions}

We prove that the $O(\Delta)$-time algorithm~\cite{astrand10vc-sc} for maximal fractional matchings is optimal:
\begin{theorem} \label{thm:main}
There is no $\LOCAL$ algorithm that finds a maximal \FM in $o(\Delta)$ rounds.
\end{theorem}
To our knowledge, this is the first linear-in-$\Delta$ lower bound in the $\LOCAL$ model for a classical graph problem. Indeed, prior lower bounds have typically fallen in one of the following categories:
\begin{itemize}[label=$-$,noitemsep]
    \item they are logarithmic in $\Delta$ \cite{kuhn04what,kuhn06price,kuhn10local},
    \item they analyse the complexity as a function of $n$ for a fixed $\Delta$ \cite{czygrinow08fast,floreen11max-min-lp,goos12local-approximation,goos12bipartite-vc,lenzen08leveraging,linial92locality,naor95what},
    \item they only hold in a model that is strictly weaker than $\LOCAL$ \cite{hirvonen12maximal-matching,kuhn06complexity}.
\end{itemize}

We hope that our methods can eventually be extended to analyse algorithms (e.g., for maximal matching) whose running times depend mildly on $n$.

\subsection{The \texorpdfstring{$\LOCAL$}{LOCAL} model}\label{ssec:localmodel}

Our result holds in the standard $\LOCAL$ model of distributed computing~\cite{linial92locality,peleg00distributed}. For now, we only recall the basic setting; see Section~\ref{sec:tools} for precise definitions.

In the $\LOCAL$ model an input graph $G = (V,E)$ defines both the problem instance and the structure of the communication network. Each node $v \in V$ is a computer and each edge $\{u,v\} \in E$ is a communication link through which nodes $u$ and $v$ can exchange messages. Initially, each node is equipped with a \emph{unique identifier} and, if we study randomised algorithms, a source of randomness. In each \emph{communication round}, each node in parallel (1)~sends a message to each neighbour, (2)~receives a message from each neighbour, and (3)~updates its local state. Eventually, all nodes have to stop and announce their local outputs---in our case the local output of a node $v \in V$ is an encoding of the weight $y(e)$ for each edge $e$ incident to~$v$. The \emph{running time} $t$ of the algorithm is the number of communication rounds until all nodes have stopped. We call an algorithm \emph{strictly local}, or simply \emph{local}, if $t=t(\Delta)$ is only a function of $\Delta$, i.e., independent of $n$.

The $\LOCAL$ model is the strongest model commonly in use---in particular, the size of each message and the amount of local computation in each communication round is unbounded---and this makes \emph{lower bounds} in this model very widely applicable.

\section{Overview}

The maximal \FM problem is an example of a \emph{locally checkable} problem: there is a local algorithm that can check whether a proposed function $y$ is a feasible solution.

It is known that randomness does not help a local algorithm in solving a locally checkable problem~\cite{naor95what}: if there is a $t(\Delta)$-time worst-case randomised algorithm, then there is a $t(\Delta)$-time deterministic algorithm (see Appendix~\ref{app:randomness}). Thus, we need only prove our lower bound for deterministic algorithms.

\subsection{Deterministic models} \label{ssec:models}

Our lower bound builds on a long line of prior research. During the course of the proof, we will visit each of the following deterministic models (see Figure~\ref{fig:models}), whose formal definitions are given in Section~\ref{sec:tools}.
\begin{itemize}[align=left,labelindent=0ex,labelwidth=5ex,labelsep*=0ex,leftmargin=5ex]
    \item[$\boldsymbol\ID$:] \emph{Deterministic $\LOCAL$.} Each node has a unique identifier \cite{peleg00distributed,linial92locality}. This is the standard model in the field of deterministic distributed algorithms.
    \item[$\boldsymbol\OI$:] \emph{Order-invariance.} The output of an algorithm is not allowed to change if we relabel the nodes while preserving the relative order of the labels \cite{naor95what}. Equivalently, the algorithm can only compare the identifiers, not access their numerical value.
    \item[$\boldsymbol\PO$:] \emph{Port numbering and orientation.} For each node, there is an ordering on the incident edges, and all edges carry an orientation \cite{mayer95local}.
    \item[$\boldsymbol\EC$:] \emph{Edge colouring.} A proper edge colouring with $O(\Delta)$ colours is given \cite{hirvonen12maximal-matching}.
\end{itemize}

The models are listed here roughly in the order of decreasing strength. For example, the $\ID$ model is strictly stronger than $\OI$, which is strictly stronger than $\PO$. However, the $\EC$ model is not directly comparable: there are problems that are trivial to solve in $\ID$, $\OI$, and $\PO$ but impossible to solve in $\EC$ with any deterministic algorithm (example: graph colouring in $1$-regular graphs); there are also problems that can be solved with a local algorithm in $\EC$ but they do not admit a local algorithm in $\ID$, $\OI$, or $\PO$ (example: maximal matching).

\begin{figure}[t]
    \centering
    \includegraphics[page=2]{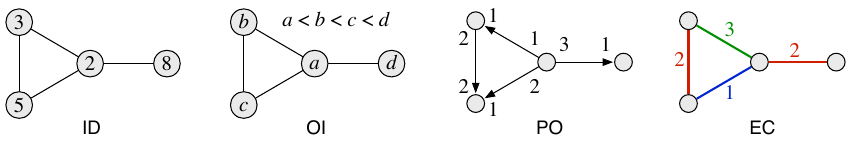}
    \caption{Deterministic models that are discussed in this work.}\label{fig:models}
\end{figure}

\subsection{Proof outline}

In short, our proof is an application of techniques that were introduced in two of our earlier works~\cite{hirvonen12maximal-matching,goos12local-approximation}. Accordingly, our proof is in two steps.

\paragraph{A weak lower bound.}

In our prior work \cite{hirvonen12maximal-matching} we showed that \emph{maximal matchings} cannot be computed in time $o(\Delta)$ in the weak $\EC$ model. The lower-bound construction there is a regular graph, and as such, tells us very little about the fractional matching problem, since maximal fractional matchings are trivial to compute in regular graphs.

Nevertheless, we use a similar \emph{unfold-and-mix} argument on what will be called \emph{loopy $\EC$-graphs} to prove the following intermediate result in Section~\ref{sec:lb-in-ec}:
\begin{step} \label{step:one}
The maximal \FM problem cannot be solved in time $o(\Delta)$ on loopy $\EC$-graphs.
\end{step}
The proof heavily exploits the limited symmetry breaking capabilities of the $\EC$ model. To continue, we need to argue that similar limitations exist in the $\ID$ model.

\paragraph{Strengthening the lower bound.}
To extend the lower bound to the $\ID$ model, we give a series of local simulation results
\[
\EC \leadsto \PO \leadsto \OI \leadsto \ID,
\]
which state that a local algorithm for the maximal fractional matching problem in one model can be simulated fast in the model preceding it. That is, even though the models $\EC$, $\PO$, $\OI$, and $\ID$ are generally very different, we show that the models are roughly equally powerful for computing a maximal fractional matching.

This part of the argument applies ideas from another prior work~\cite{goos12local-approximation}. There, we showed that, for a large class of optimisation problems, a run-time preserving simulation $\PO\leadsto\ID$ exists. Unfortunately, the maximal fractional matching problem is not included in the scope of this result (fractional matchings are not \emph{simple} in the sense of~\cite{goos12local-approximation}), so we may not apply this result directly in a black-box fashion. In addition, this general result does not hold for the $\EC$ model.

Nevertheless, we spend Section~\ref{sec:simulations} extending the methods of~\cite{goos12local-approximation} and show that they can be tailored to the case of fractional matchings:
\begin{step} \label{step:two}
If the maximal \FM problem can be solved in time $t(\Delta)$ on $\ID$-graphs, then it can be solved in time $t(\Theta(\Delta))$ on loopy $\EC$-graphs.
\end{step}

In combination with Step~\ref{step:one}, this proves Theorem~\ref{thm:main}.

\section{Tools of the Trade} \label{sec:tools}

Before we dive into the lower-bound proof, we recall the definitions of the four models mentioned in Section~\ref{ssec:models}, and describe the standard tools that are used in their analysis.

\subsection{Locality}

Distributed algorithms are typically described in terms of networked state machines: the nodes of a network exchange messages for $t$ synchronous communication rounds after which they produce their local outputs (cf.\ Section~\ref{ssec:localmodel}).

Instead, for the purposes of our lower-bound analysis, we view an algorithm $\Alg$ simply as a function that associates to each pair $(G,v)$ an output $\Alg(G,v)$ in a way that respects 
\emph{locality}. That is, an algorithm $\Alg$ is said to have run-time $t$, if the output $\Alg(G,v)$ depends only on the information that is available in the radius-$t$ neighbourhood around $v$. More formally, define
\[
\tau_t(G,v) \subseteq (G,v)
\]
as consisting of the nodes and edges of $G$ that are within distance $t$ from $v$---the distance of an edge $\{u,w\}$ from $v$ is $\min\{\dist(v,u),\dist(v,w)\}+1$. A $t$-time algorithm $\Alg$ is then a mapping that satisfies
\begin{equation} \label{eq:locality}
\Alg(G,v) = \Alg(\tau_t(G,v)).
\end{equation}

The information contained in $\tau_t(G,v)$ depends on which of the models $\EC$, $\PO$, $\OI$, and $\ID$ we are studying. For each model we define an associated graph class.

\subsection{Identifier-based networks}

\paragraph{\boldmath$\ID$-graphs.}
An \emph{$\ID$-graph} is simply a graph $G$ whose nodes are assigned unique identifiers; namely, $V(G)\subseteq \N$. Any mapping $\Alg$ satisfying (\ref{eq:locality}) is a $t$-time $\ID$-algorithm.

\paragraph{\boldmath$\OI$-graphs.}
An \emph{$\OI$-graph} is an ordered graph $(G,\preceq)$ where $\preceq$ is a linear order on $V(G)$. An $\OI$-algorithm $\Alg$ operates on $\OI$-graphs in such a way that if $(G,\preceq,v)$ and $(G',\preceq',v')$ are isomorphic (as ordered structures), then $\Alg(G,\preceq,v)=\Alg(G',\preceq',v')$.

\bigskip
Every $\ID$-graph $G$ is naturally an $\OI$-graph $(G,\leq)$ under the usual order $\leq$ on $\N$. In the converse direction, we often convert an $\OI$-graph $(G,\preceq)$ into an $\ID$-graph by specifying an $\ID$-assignment $\varphi\colon V(G)\to\N$ that \emph{respects} $\preceq$ in the sense that $v\preceq u$ implies $\varphi(v)\leq \varphi(u)$. The resulting $\ID$-graph is denoted $\varphi(G)$.

\subsection{Anonymous networks} \label{ssec:anonymous}

On anonymous networks the nodes do not have identifiers. The only symmetry breaking information is now provided in an \emph{edge colouring} of a suitable type. This means that whenever there is an isomorphism between $(G,v)$ and $(G',v')$ that preserves edge colours, we will have $\Alg(G,v)=\Alg(G',v')$.

\paragraph{\boldmath$\EC$-graphs.} An \emph{$\EC$-graph} carries a proper edge colouring $E(G)\to \{1,\ldots,k\}$, where $k=O(\Delta)$. That is, if two edges are adjacent, they have distinct colours.

\paragraph{\boldmath$\PO$-graphs.} A \emph{$\PO$-graph} is a directed graph whose edges are coloured in the following way: if $(u,v)$ and $(u,w)$ are outgoing edges incident to $u$, then they have distinct colours; and if $(v,u)$ and $(w,u)$ are incoming edges incident to $u$, then they have distinct colours. Thus, we may have $(v,u)$ and $(u,w)$ coloured the same.

\bigskip

We find it convenient to treat $\PO$-graphs as edge-coloured digraphs, even if this view is nonstandard. Usually, $\PO$-graphs are defined as digraphs with a \emph{port numbering}, i.e., each node is given an ordering of its neighbours. This is equivalent to our definition: A port numbering gives rise to an edge colouring where an edge $(u,v)$ is coloured with $(i,j)$ if $v$ is the $i$-th neighbour of $u$ and $u$ is the $j$-th neighbour of $v$ (see Figure~\ref{fig:podef}a). Conversely, we can derive a port numbering from an edge colouring---first take all outgoing edges ordered by the edge colours, and then take all incoming edges ordered by the edge colours (Figure~\ref{fig:podef}b).

\begin{figure}
    \centering
    \includegraphics[page=8]{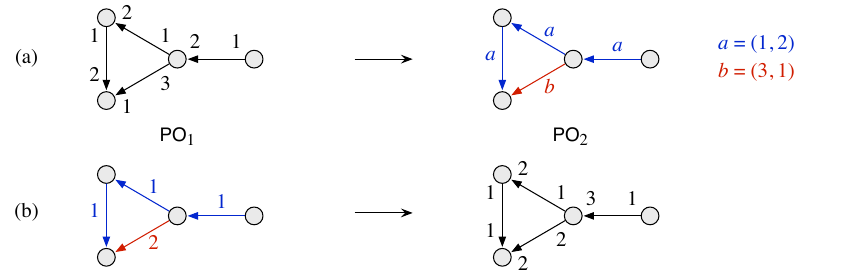}
    \caption{Two equivalent definitions of $\PO$-graphs: ($\PO_1$) a node of degree $d$ can refer to incident edges with labels $1,2,\dotsc,d$; ($\PO_2$) edges are coloured so that incoming edges have distinct colours and outgoing edges have distinct colours.}\label{fig:podef}
\end{figure}

We are not done with defining $\EC$ and $\PO$ algorithms. We still need to restrict their power by requiring that their outputs are \emph{invariant under graph lifts}, as defined next.

\subsection{Lifts}

A graph $H$ is said to be a \emph{lift} of another graph $G$ if there exists an onto graph homomorphism $\alpha\colon V(H)\to V(G)$ that is a \emph{covering map}, i.e., $\alpha$ preserves node degrees, $\deg_H(v) = \deg_G(\alpha(v))$. Our discussion of lifts always takes place in either $\EC$ or $\PO$; in this context we require that a covering map preserves edge colours.
\begin{center}
    \includegraphics[page=9]{figs.pdf}
\end{center}

The defining characteristic of anonymous models is that the output of an algorithm is invariant under taking lifts. That is, if $\alpha\colon V(H)\to V(G)$ is a covering map, then
\begin{equation} \label{eq:lift}
\Alg(H,v) = \Alg(G,\alpha(v)),\qquad\text{for each}\ v\in V(H).
\end{equation}
Since an isomorphism between $H$ and $G$ is a special case of a covering map, the condition~(\ref{eq:lift}) generalises the discussion in Section~\ref{ssec:anonymous}. We will be exploiting this limitation extensively in analysing the models $\EC$ and $\PO$.

Graphs are partially ordered by the \emph{lift} relation. For any connected graph $G$, there are two graphs $U_G$ and $F_G$ of special interest that are related to $G$ via lifts.

\paragraph{Universal cover \boldmath$U_G$.}
The \emph{universal cover} $U_G$ of $G$ is an unfolded tree-like version of~$G$. More precisely, $U_G$ is the unique tree that is a lift of $G$. Thus, if $G$ is a tree, $U_G = G$; if $G$ has cycles, $U_G$ is infinite. In passing from $G$ to $U_G$ we lose all the cycle structure that is present in $G$. The universal cover is often used to model the information that a distributed algorithm---even with unlimited running time---is able to collect on an anonymous network~\cite{angluin80local}.

\begin{center}
    \includegraphics[page=10]{figs.pdf}
\end{center}

\paragraph{Factor graph \boldmath$F_G$.}
The \emph{factor graph} $F_G$ of $G$ is the smallest graph $F$ such that $G$ is a lift of $F$; see Figure~\ref{fig:factorgraph}. In general, $F_G$ is a multigraph with loops and parallel edges. It is the most concise representation of all the global symmetry breaking information available in $G$. For example, in the extreme case when $G$ is vertex-transitive, $F_G$ consists of just one node and some loops.

\begin{figure}[b]
    \centering
    \includegraphics[page=11]{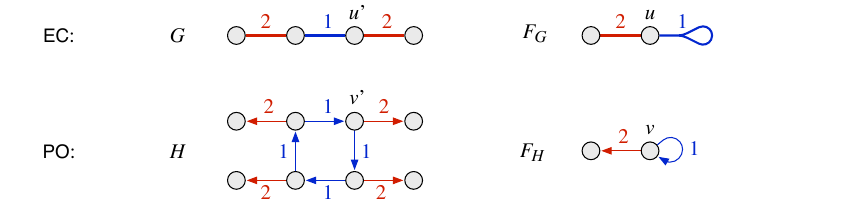}
    \caption{Factor graphs and loops. We follow the convention that undirected loops in $\EC$-graphs count as a single incident edge, while directed loops in $\PO$-graphs count as two incident edges: an incoming edge and an outgoing edge. In this example, both $u$ and its preimage $u'$ are nodes of degree $2$; they are incident to one edge of colour $1$ and one edge of colour $2$. Both $v$ and its preimage $v'$ are nodes of degree $3$; they are incident to two outgoing edges of colours $1$ and $2$, and one incoming edge of colour $1$.}\label{fig:factorgraph}
\end{figure}

Even though we want our input graphs always to be simple, we may still analyse $\EC$ and $\PO$-algorithms $\Alg$ on multigraphs $F$ with the understanding that the output $\Alg(F,v)$ is interpreted according to~(\ref{eq:lift}). That is, to determine $\Alg(F,v)$, do the following:
\begin{enumerate}[noitemsep]
\item Lift $F$ to a simple graph $G$ (e.g., take $G=U_F$) via some $\alpha\colon V(G)\to V(F)$.
\item Execute $\Alg$ on $(G,u)$ for some $u\in\alpha^{-1}(v)$.
\item Interpret the output of $u$ as an output of $v$.
\end{enumerate}
In what follows we refer to multigraphs simply as graphs.

\subsection{Loops} \label{sec:loops}

In $\EC$-graphs, a single loop on a node contributes $+1$ to its degree, whereas in $\PO$-graphs, a single (directed) loop contributes $+2$ to the degree, once for the tail and once for the head. This is reflected in the way we draw loops---see Figure~\ref{fig:factorgraph}.

The loop count on a node $v\in V(G)$ measures the inability of $v$ to break local symmetries. Indeed, if $v$ has $\ell$ loops, then in any simple lift $H$ of $G$ each node $u\in V(H)$ that is mapped to $v$ by the covering map will have $\ell$ distinct neighbours $w_1,\ldots,w_\ell$ that, too, get mapped to $v$. Thus, an anonymous algorithm is forced to output the same on $u$ as on each of $w_1,\ldots,w_\ell$.

We consider loops as an important resource.
\begin{definition}[Loopiness]
An edge-coloured graph $G$ is called \emph{$k$-loopy} if each node in $F_G$ has at least $k$ loops. A graph is simply \emph{loopy} if it is $1$-loopy.
\end{definition}

When computing maximal fractional matchings on a loopy graph $G$, an anonymous algorithm must saturate all the nodes. For suppose not. If $v\in V(G)$ is a node that does not get saturated, the loopiness of $G$ implies that $v$ has a neighbour $u$ (can be $u=v$ via a loop) that produces the same output as $v$. But now neither endpoint of $\{u,v\}$ is saturated, which contradicts maximality; see Figure~\ref{fig:ec-saturate}. We record this observation.

\begin{figure}
    \centering
    \includegraphics[page=13]{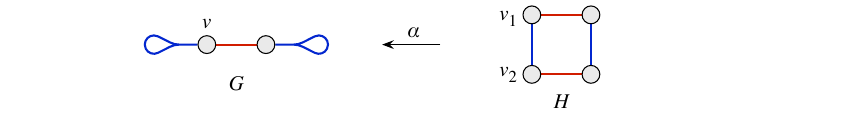}
    \caption{$\EC$-graph $G$ is loopy. Assume that an $\EC$-algorithm $\Alg$ produces an output in which node $v$ is unsaturated. Then we can construct a simple $\EC$-graph $H$ that is a lift of $G$ via $\alpha\colon V(H) \to V(G)$ such that $\alpha(v_1) = \alpha(v_2) = v$ and $\{v_1,v_2\} \in E(H)$. If we apply $\Alg$ to $H$, both $v_1$ and $v_2$ are unsaturated; hence $\Alg$ fails to produce a maximal \FM.}\label{fig:ec-saturate}
\end{figure}

\begin{lemma}\label{lem:ec-saturate}
Any $\EC$-algorithm for the maximal \FM problem computes a fully saturated \FM on a loopy $\EC$-graph.\qed
\end{lemma}

\section{Lower Bound in \texorpdfstring{\boldmath$\EC$}{EC}} \label{sec:lb-in-ec}

In this section we carry out Step~\ref{step:one} of our lower-bound plan. To do this we extend the previous lower bound result~\cite{hirvonen12maximal-matching} to the case of maximal fractional matchings.

\subsection{Strategy}

Let $\Alg$ be any $\EC$-algorithm computing a maximal fractional matching. We construct inductively a sequence of $\EC$-graph pairs
\[
(G_i,H_i),\quad i=0,1,\ldots,\Delta-2,
\]
that witness $\Alg$ having run-time greater than $i$. Each of the graphs $G_i$ and $H_i$ will have maximum degree at most $\Delta$, so for $i=\Delta-2$, we will have the desired lower bound. More precisely, we show that there are nodes $g_i\in V(G_i)$ and $h_i\in V(H_i)$ satisfying the following property:
\begin{enumerate}[label=(P\arabic*)]
\item The $i$-neighbourhoods $\tau_i(G_i,g_i)$ and $\tau_i(H_i,h_i)$ are isomorphic---yet,
\[
\Alg(G_i,g_i) \neq \Alg(H_i,h_i).
\]
Moreover, there is a loop of some colour $c_i$ adjacent to both $g_i$ and $h_i$ such that the outputs disagree on its weight.
\end{enumerate}
We will also make use of the following additional properties in the construction:
\begin{enumerate}[resume*]
\item The graphs $G_i$ and $H_i$ are $(\Delta-1-i)$-loopy---consequently, $\Alg$ will saturate all their nodes by Lemma~\ref{lem:ec-saturate}.
\item When the loops are ignored, both $G_i$ and $H_i$ are trees.
\end{enumerate}

\subsection{Base case \texorpdfstring{$(i = 0)$}{(i = 0)}}

Let $G_0$ consist of a single node $v$ that has $\Delta$ differently coloured loops. When $\Alg$ is run on $G_0$, it saturates $v$ by assigning at least one loop $e$ a non-zero weight; see Figure~\ref{fig:ecbasecase}. Letting $H_0 := G_0-e$ it is now easy to check that the pair $(G_0,H_0)$ satisfies (P1--P3) for $g_0=h_0=v$. For example, $\tau_0(G_0,v)\cong\tau_0(H_0,v)$ only because we consider the loops to be at distance $1$ from $v$.

\begin{figure}
    \centering
    \includegraphics[page=3]{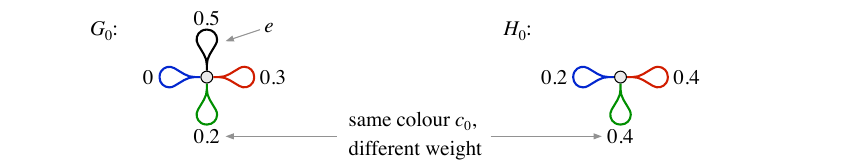}
    \caption{Base case. By removing a loop $e$ with a non-zero weight, we force the algorithm to change the weight of at least one edge that is present in both $G_0$ and~$H_0$.}\label{fig:ecbasecase}
\end{figure}

\subsection{Inductive step}

Suppose $(G_i,H_i)$ is a pair satisfying (P1--P3). For convenience, we write $G$, $H$, $g$, $h$, and $c$ in place of $G_i$, $H_i$, $g_i$, $h_i$, and $c_i$. Also, we let $e\in E(G)$ and $f\in E(H)$ be the colour-$c$ loops adjacent to $g$ and $h$ to which $\Alg$ assigns different weights.

To construct the pair $(G_{i+1},H_{i+1})$ we unfold and mix; see Figure~\ref{fig:ecunfoldmix}

\paragraph{Unfolding.}
First, we unfold the loop $e$ in $G$ to obtain a 2-lift $\GG$ of $G$. That is, $\GG$ consists of two disjoint copies of $G-e$ and a new edge of colour $c$ (which we still call $e$) that connects the two copies of $g$ in $\GG$. For notational purposes, we fix some identification $V(G)\subseteq V(\GG)$ so that we can easily talk about one of the copies. Similarly, we construct a 2-lift $\HH$ of $H$ by unfolding the loop $f$.

Recall that $\Alg$ cannot tell apart $G$ from $\GG$, or $H$ from $\HH$. In particular $\Alg$ continues to assign unequal weights to $e$ and $f$ in these lifts.

\begin{figure}
    \centering
    \includegraphics[page=4]{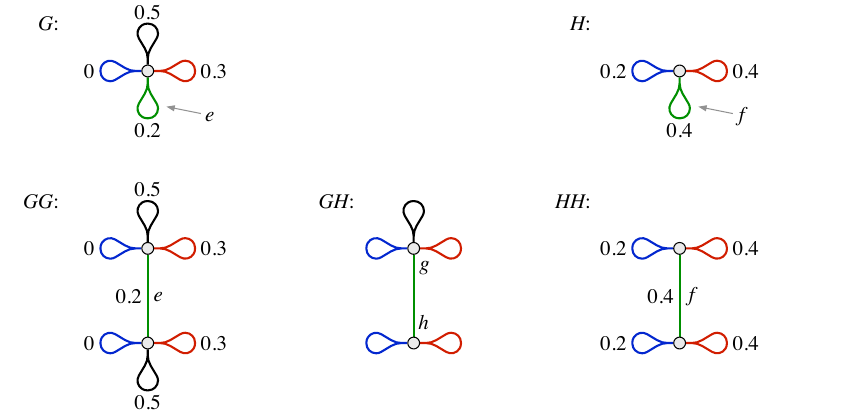}
    \caption{Unfold and mix. The weights of $e$ and $f$ differ; hence the weight of $\{g,h\}$ is different from the weight of $e$ or $f$.}\label{fig:ecunfoldmix}
\end{figure}

\paragraph{Mixing.}
Next, we mix together the graphs $\GG$ and $\HH$ to obtain a graph $\GH$ defined as follows: $\GH$ contains a copy of $G-e$, a copy of $H-f$, and a new colour-$c$ edge that connects the nodes $g$ and $h$. For notational purposes, we let $V(\GH) := V(G)\cup V(H)$, where we tacitly assume that $V(G)\cap V(H) = \varnothing$.

\paragraph{Analysis.}
Consider the weight that $\Alg$ assigns to the colour-$c$ edge $\{g,h\}$ in $\GH$.
Since $\Alg$ gives the edges $e$ and $f$ different weights in $\GG$ and $\HH$, we must have that the weight of $\{g,h\}$ differs from the weight of $e$ or the weight of $f$ (or both). We assume the former (the latter case is analogous), and argue that the pair
\[
(G_{i+1},H_{i+1}) := (\GG,\GH)
\]
satisfies the properties (P1--P3). It is easy to check that (P2) and (P3) are satisfied by the construction; it remains is to find the nodes $g_{i+1}\in V(\GG)$ and $h_{i+1}\in V(\GH)$ that satisfy (P1).

To this end, we exploit the following property of fractional matchings:
\begin{fact}[Propagation principle]
Let $y$ and $y'$ be fractional matchings that saturate a node $v$. If $y$ and $y'$ disagree on some edge incident to $v$, there must be another edge incident to $v$ where $y$ and $y'$ disagree.
\end{fact}
Our idea is to apply this principle in a fully saturated graph, where the disagreements propagate until they are resolved at a loop; this is where we locate $g_{i+1}$ and $h_{i+1}$. See Figure~\ref{fig:ecpropagation} for an example.

We consider the following fully saturated fractional matchings on~$G$:
\begin{align*}
y\  &=\ \text{the \FM determined by $\Alg$'s output on the nodes $V(G)$ in $\GG$}, \\[-3pt]
y'\ &=\ \text{the \FM determined by $\Alg$'s output on the nodes $V(G)$ in $\GH$}.
\end{align*}
Starting at the node $g\in V(G)$ we already know by assumption that $y$ and $y'$ disagree on the colour-$c$ edge incident to $g$. Thus, by the propagation principle, $y$ and $y'$ disagree on some other edge incident to $g$. If this edge is not a loop, it connects to a neighbour $g'\in V(G)$ of $g$ and the argument can be continued: because $y$ and $y'$ disagree on $\{g,g'\}$, there must be another edge incident to $g'$ where $y$ and $y'$ disagree, and so on. Since $G$ does not have any cycles (apart from the loops), this process has to terminate at some node $g^*\in V(G)$ such that $y$ and $y'$ disagree on a loop $e^*\neq e$ incident to $g^*$. Note that $e^*$ is a loop in both $\GG$ and $\GH$, too. Thus, we have found our candidate $g_{i+1} = h_{i+1} = g^*$.

\begin{figure}
    \centering
    \includegraphics[page=5]{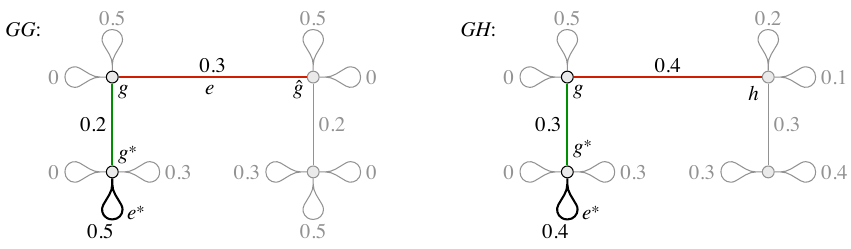}
    \caption{Propagation. The weights of $e$ and $\{g,h\}$ differ. We apply the propagation principle towards the common part $G$ that is shared by $\GG$ and $\GH$. The graphs are loopy and hence all nodes are saturated by $\Alg$; we will eventually find a loop $e^*$ that is present in both $\GG$ and $\GH$, with different weights.}\label{fig:ecpropagation}
\end{figure}

To finish the proof, we need to show that
\begin{equation} \label{eq:neigh-cong}
\tau_{i+1}(\GG,g^*) \cong \tau_{i+1}(\GH,g^*).
\end{equation}
The critical case is when $g^* = g$ as this node is the closest among $V(G)$ to seeing the topological differences between the graphs $\GG$ and $\GH$. Starting from $g$ and stepping along the colour-$c$ edge towards the differences, we arrive, in $\GG$, at a node $\hat{g}$ that is a copy of $g\in V(G)$, and in $\GH$, at the node $h$. But these nodes satisfy
\[
\tau_i(\GG,\hat{g})\cong \tau_i(\GH,h)
\]
by our induction assumption. Using this, (\ref{eq:neigh-cong}) follows.

\section{Local Simulations} \label{sec:simulations}

Now that we have an $\Omega(\Delta)$ time lower bound in the $\EC$ model, our next goal is to extend this result to the $\ID$ model. In this section we implement Step~\ref{step:two} of our plan and give a series of local simulations
\[
\EC \leadsto \PO \leadsto \OI \leadsto \ID.
\]
Here, each simulation preserves the running time of an algorithm up to a constant factor. In particular, together with Step~\ref{step:one}, this will imply the $\Omega(\Delta)$ time lower bound in the $\ID$ model.

\subsection{Simulation \texorpdfstring{$\EC\leadsto\PO$}{EC to PO}} \label{sec:ec-simulates-po}

We start with the easiest simulation.
Suppose there is a $t$-time $\PO$-algorithm for the maximal fractional matching problem on graphs of maximum degree $\Delta$; we describe a $t$-time $\EC$-algorithm for graphs of maximum degree $\Delta/2$.

The local simulation is simple; see Figure~\ref{fig:ecpo}. On input an $\EC$-graph $G$ we interpret each edge $\{u,v\}$ of colour $c$ as two directed edges $(u,v)$ and $(v,u)$, both of colour~$c$; this interpretation makes $G$ into a $\PO$-graph $G_\leftrightarrows$. We can now locally simulate the $\PO$-algorithm on $G_\leftrightarrows$ to obtain an \FM $y$ as output. Finally, we transform $y$ back to an \FM of $G$: the edge $\{u,v\}$ is assigned weight $y(u,v) + y(v,u)$.

\begin{figure}
    \centering
    \includegraphics[page=7]{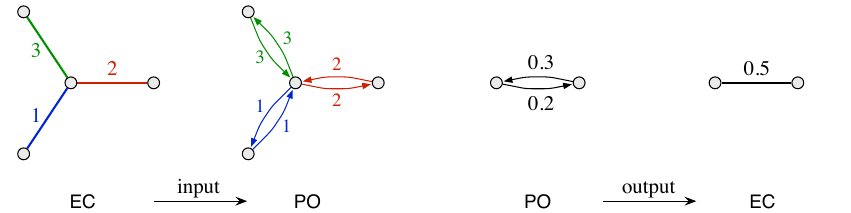}
    \caption{$\EC\leadsto\PO$. Mapping an $\EC$-graph $G$ into a $\PO$-graph $G_\leftrightarrows$, and mapping the output of a $\PO$-algorithm back to the original graph.}\label{fig:ecpo}
\end{figure}

\subsection{Tricky identifiers} \label{sec:tricky-ids}

When we are computing a maximal fractional matching $y\colon E(G)\to[0,1]$, we have, a priori, infinitely many choices for the weight $y(e)$ of an edge. For example, in a path on nodes $v_1$, $v_2$, and $v_3$, we can freely choose $y(\{v_1,v_2\})\in[0,1]$ provided we set $y(\{v_2,v_3\})=1-y(\{v_1,v_2\})$. In particular, an $\ID$-algorithm can output edge weights that depend on the node identifiers whose magnitude is not bounded.

Unbounded outputs are tricky from the perspective of proving lower bounds. The main result of the recent work~\cite{goos12local-approximation} is a run-time preserving local simulation $\PO\leadsto \ID$, but the result only holds under the assumption that the solution can be encoded using finitely many values per node on graphs of maximum degree $\Delta$. This restriction has its source in an earlier local simulation $\OI\leadsto\ID$ due to Naor and Stockmeyer~\cite{naor95what} that is crucially using Ramsey's theorem. In fact, these two local simulation results fail if unbounded outputs are allowed; counterexamples include even natural graph problems~\cite{hasemann12scheduling}.

In conclusion, we need an ad hoc argument to establish that an $\ID$-algorithm cannot benefit from unique identifiers in case of the maximal fractional matching problem.

\subsection{Simulation \texorpdfstring{$\PO\leadsto\OI$}{PO to OI}} \label{sec:po-simulates-oi}

Before we address the question of simulating $\ID$-algorithms, we first salvage one part of the result in~\cite{goos12local-approximation}: there is local simulation $\PO\leadsto\OI$ that applies to many locally checkable problems, regardless of the size of the output encoding. Even though this simulation works off-the-shelf in our present setting, we cannot use this result in a black-box fashion, as we need to access its inner workings later in the analysis. Thus, we proceed with a self-contained proof.

The following presentation is considerably simpler than that in~\cite{goos12local-approximation}, since we are only interested in a simulation that produces a \emph{locally maximal} fractional matching, not in a simulation that also provides approximation guarantees on the \emph{total weight}, as does the original result.

\paragraph{\boldmath$\PO$-checkability.}
Maximal fractional matchings are not only locally checkable, but also \emph{$\PO$-checkable}: there is a local $\PO$-algorithm that can check whether a given $y$ is a maximal \FM. An important consequence of $\PO$-checkability is that if $H$ is a lift of $G$ then any $\PO$-algorithm produces a feasible solution on $H$ if and only if it produces a feasible solution on $G$.

\paragraph{Order homogeneity.}
The key to the simulation $\PO\leadsto\OI$ is a \emph{canonical linear order} that can be computed for any tree-like $\PO$-neighbourhood. To define this ordering, let $d$ denote the maximum number of edge colours appearing in the input $\PO$-graphs that have maximum degree $\Delta$, and let $T$ denote the infinite $2d$-regular $d$-edge-coloured $\PO$-tree. We fix a homogeneous linear order for $T$:
\begin{lemma} \label{lem:tree-order}
There is a linear order $\preceq$ on $V(T)$ such that all the ordered neighbourhoods $(T,\preceq,v)$, $v\in V(T)$, are pairwise isomorphic (i.e., up to any radius).
\end{lemma}
For a proof, see Appendix~\ref{app:tree-order}.

\paragraph{Simulation.}
Let $\Alg_\OI$ be any $t$-time $\OI$-algorithm solving a $\PO$-checkable problem; we describe a $t$-time $\PO$-algorithm $\Alg_\PO$ solving the same problem.

The algorithm $\Alg_\PO$ operates on a $\PO$-graph $G$ as follows; see Figure~\ref{fig:po-simulates-oi}. Given a $\PO$-neighbourhood $\tau:=\tau_t(U_G,v)$, we first embed $\tau$ in $T$: we choose an arbitrary node $u\in V(T)$, identify $v$ with $u$, and let the rest of the embedding $\tau\subseteq (T,u)$ be dictated uniquely by the edge colours. We then use the ordering $\preceq$ inherited from $T$ to order the nodes of $\tau$. By Lemma~\ref{lem:tree-order}, the resulting structure $(\tau,\preceq)$ is independent of the choice of $u$, i.e., the isomorphism type of $(\tau,\preceq)$ is only a function of $\tau$. Finally, we simulate
\begin{equation} \label{eq:po-simulates-oi}
\Alg_\PO(\tau) := \Alg_\OI(\tau,\preceq).
\end{equation}

\begin{figure}
    \centering
    \includegraphics[page=14]{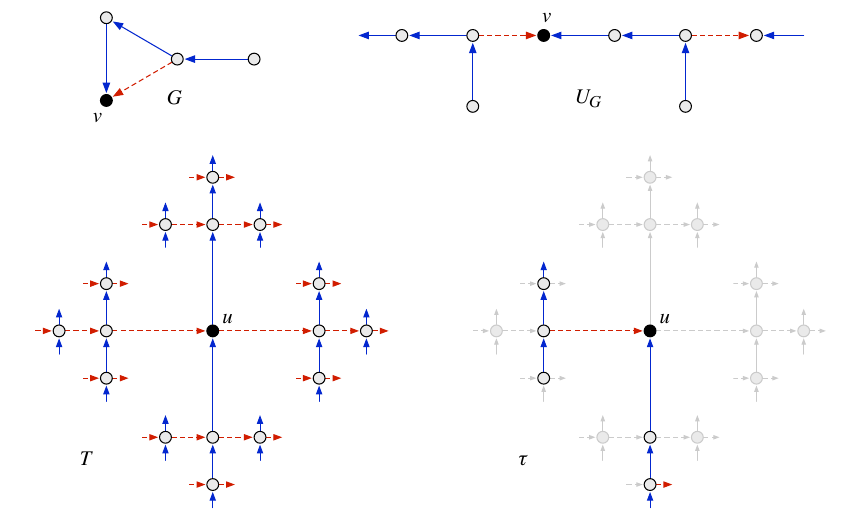}
    \caption{Given a $\PO$-graph $G$, algorithm $\Alg_\PO$ simulates the execution of $\Alg_\OI$ on $\OI$-graph $\tau$. The linear order on $V(\tau)$ is inherited from the regular tree $T$. As $T$ is homogeneous, the linear order does not depend on the choice of node $u$ in $T$.}\label{fig:po-simulates-oi}
\end{figure}

To see that the output of $\Alg_\PO$ is feasible, we argue as follows. Embed the universal cover $U_G$ as a subgraph of $(T,\preceq)$ in a way that respects edge colours. Again, all possible embeddings are isomorphic; we call the inherited ordering $(U_G,\preceq)$ the \emph{canonical ordering} of $U_G$. Our definition of $\Alg_\PO$ and the order homogeneity of $(T,\preceq)$ now imply that
\[
\Alg_\PO(U_G,v) = \Alg_\OI(U_G,\preceq,v)\qquad \text{for all}\ v\in V(U_G).
\]
Therefore, the output of $\Alg_\PO$ is feasible on $U_G$. Finally, by $\PO$-checkability, the output of $\Alg_\PO$ is feasible also on $G$, as desired.

\subsection{Simulation \texorpdfstring{$\OI\leadsto\ID$}{OI to ID}} \label{sec:oi-simulates-id}

The reason why an $\ID$-algorithm $\Alg$ cannot benefit from unbounded identifiers is due to the propagation principle. We formalise this in two steps.
\begin{enumerate}[label=(\roman*)]
\item We use the Naor--Stockmeyer $\OI\leadsto\ID$ result to see that $\Alg$ can be forced to output fully saturated \FM{}s on so-called \emph{loopy} $\OI$-neighbourhoods.
\item We then observe that, on these neighbourhoods, $\Alg$ behaves like an $\OI$-algorithm: $\Alg$'s output cannot change if we relabel a node in an order-preserving fashion, because the changes in the output would have to propagate outside of $\Alg$'s run-time.
\end{enumerate}
That is, our simulation $\OI\leadsto\ID$ will work only on certain types of neighbourhoods (in contrast to our previous simulations), but this will be sufficient for the purposes of the lower bound proof.

\paragraph{Step (i).} Let $\Alg$ be a $t$-time $\ID$-algorithm that computes a maximal fractional matching on graphs of maximum degree $\Delta$.

From $\Alg$ we can derive, by a straightforward simulation, a $t$-time \emph{binary-valued} $\ID$-algorithm $\Alg^*$ that indicates whether $\Alg$ saturates a node. That is, $\Alg^*(G,v) := 1$ if $\Alg$ saturates $v$ in $G$, otherwise $\Alg^*(G,v):=0$. Such saturation indicators $\Alg^*$ were considered previously in~\cite[\defaultS4]{astrand09vc2apx}.

Because (and \emph{only} because) $\Alg^*$ outputs finitely many values, we can now apply the Ramsey technique of Naor and Stockmeyer~\cite[Lemma 3.2]{naor95what}. To avoid notational clutter, we use a version of their result that follows from the application of the infinite Ramsey's theorem (rather than the finite):
\begin{lemma}[Naor and Stockmeyer] \label{lem:ramsey}
There is an infinite set $I\subseteq\N$ such that $\Alg^*$ is an $\OI$-algorithm when restricted to graphs whose identifiers are in $I$. \qed
\end{lemma}

We say that $\tau_t(U_G,\preceq,v)$ is a \emph{loopy} $\OI$-neighbourhood if $G$ is a loopy $\PO$-graph and $(U_G,\preceq)$ is the canonically ordered universal cover of~$G$. We also denote by $B_t(v)\subseteq V(U_G)$ the node set of $\tau_t(U_G,v)$.

Our saturation indicator $\Alg^*$ is useful in proving the following lemma, which encapsulates step (i) of our argument.
\begin{lemma} \label{lem:oi-saturate}
Let $\tau:=\tau_t(U_G,\preceq,v)$ be loopy. If $\varphi\colon B_t(v)\to I$ is an $\ID$-assignment to the nodes of $\tau$ that respects $\preceq$, then $\Alg$ saturates $v$ under $\varphi$.
\end{lemma}
\begin{proof}
By loopiness of $G$, the node $v$ has a neighbour $u\in V(U_G)$ such that $\tau_t(U_G,v)\cong\tau_t(U_G,u)$ as $\PO$-neighbourhoods. By order homogeneity, $\tau_t(U_G,\preceq,v)\cong\tau_t(U_G,\preceq,v)$ as $\OI$-neighbourhoods. By Lemma \ref{lem:ramsey}, this forces $\Alg^*$ to output the same on $v$ and $u$ under any $\ID$-assignment $\varphi'\colon B_t(v)\cup B_t(u) \to I$ that respects $\preceq$. But $\Alg^*$ cannot output two adjacent $0$'s if $\Alg$ is to produce a maximal fractional matching. Hence, $\Alg^*$ outputs $1$ on $\varphi'(\tau)$. Finally, by order-invariance, $\Alg^*$ outputs $1$ on $\varphi(\tau)$, which proves the claim.
\end{proof}

\paragraph{Step (ii).}
Define $J$ as an infinite subset of $I$ that is obtained by picking every $(m+1)$-th identifier from $I$, where $m$ is the maximum number of nodes in a $(2t+1)$-neighbourhood of maximum degree $\Delta$. That is, for any two $j,j'\in J$, $j<j'$, there are $m$ distinct identifiers $i\in I$ with $j<i<j'$.

The next lemma states that $\Alg$ behaves like an $\OI$-algorithm on loopy neighbourhoods that have identifiers from $J$.
\begin{lemma} \label{lem:loopy-neigh}
Let $\tau:=\tau_t(U_G,\preceq,v)$ be loopy. If $\varphi_1,\varphi_2\colon B_t(v)\to J$ are any two $\ID$-assignments that respect $\preceq$, then $\Alg(\varphi_1(\tau)) = \Alg(\varphi_2(\tau))$.
\end{lemma}
\begin{proof} We first consider the case where $\varphi_1$ and $\varphi_2$ disagree only on a single node~$v^*\in B_t(v)$. Towards a contradiction suppose that
\begin{equation} \label{eq:contr-assumption}
\Alg(\varphi_1(\tau)) \neq \Alg(\varphi_2(\tau)).
\end{equation}

We start with partial $\ID$-assignments for $U_G$ that are defined on the nodes $B_{2t+1}(v)$; this will suffice for running $\Alg$ on the nodes $B_{t+1}(v)$. Indeed, because $J\subseteq I$ is sufficiently sparse, we can extend $\varphi_1$ and $\varphi_2$ into assignments $\bar{\varphi}_1,\bar{\varphi}_2\colon B_{2t+1}(v)\to I$ such that
\begin{itemize}[label=$-$,noitemsep]
\item $\bar{\varphi}_1$ and $\bar{\varphi}_2$ respect $\preceq$, and
\item $\bar{\varphi}_1$ and $\bar{\varphi}_2$ still disagree only on the node $v^*$.
\end{itemize}
Let $y_i$, $i=1,2$, be the fractional matching defined on the edges incident to $B_{t+1}(v)$ that is determined by the output of $\Alg$ on the nodes $B_{t+1}(v)$ under the assignment $\bar{\varphi}_i$. By Lemma~\ref{lem:oi-saturate}, all the nodes $B_{t+1}(v)$ are saturated in both $y_1$ and $y_2$.

Let $D\subseteq U_G$ be the subgraph consisting of the edges $e$ with $y_1(e)\neq y_2(e)$ and of the nodes that are incident to such edges; by (\ref{eq:contr-assumption}), we have $v\in V(D)$. Now we can reinterpret the propagation principle from Section~\ref{sec:lb-in-ec}:
\begin{fact}[Propagation principle]
Each node $u\in B_{t+1}(v)\cap V(D)$ has $\deg_D(u)\geq 2$.
\end{fact}
Using the fact that $D\subseteq U_G$ is a tree, we can start a simple walk at $v\in V(D)$, take the first step away from $v^*$, and finally arrive at a node $u\in B_{t+1}(v)\cap V(D)$ that has $\dist(u,v^*)\geq t+1$, i.e, the node $u$ does not see the difference between the assignments $\bar{\varphi}_1$ and $\bar{\varphi}_2$. But this is a contradiction: as the $t$-neighbourhoods $\bar{\varphi}_i(\tau_{t}(U_G,u))$, $i=1,2$, are the same, so should the weights output by $\Alg$.

\emph{General case.} 
If $\varphi_1,\varphi_2\colon B_t(v)\to J$ are any two assignments respecting~$\preceq$, they can be related to one another by a series of assignments 
\[
\varphi_1=\pi_1,\pi_2,\ldots,\pi_k=\varphi_2,
\]
where any two consecutive assignments $\pi_{i}$ and $\pi_{i+1}$ both respect $\preceq$ and disagree on exactly one node. Thus, the claim follows from the analysis above.
\end{proof}
Let $\Alg_\OI$ be any $t$-time $\OI$-algorithm that agrees with the order-invariant output of $\Alg$ on loopy $\OI$-neighbourhoods that have identifiers from $J$. We now obtain the final form of our $\OI\leadsto\ID$ simulation:
\begin{corollary} \label{cor:oi-simulates-id}
If $G$ is a loopy $\PO$-graph, $\Alg_\OI$ produces a maximal fractional matching on the canonically ordered universal cover $(U_G,\preceq)$.
\end{corollary}
\begin{proof}
The claim follows by a standard argument~\cite[Lemma 3.2]{naor95what} from two facts: $J$ is large enough; and maximal fractional matchings are locally checkable.
\end{proof}

\subsection{Conclusion} \label{sec:conclusion}

To get the final lower bound of Theorem~\ref{thm:main} we reason backwards. Assume that $\Alg$ is a $t$-time $\ID$-algorithm that computes a maximal fractional matching on any graph of maximum degree $\Delta$.
\begin{itemize}[leftmargin=*,label={$\EC\leadsto\PO$:}]
    \item[$\OI\leadsto\ID$:] Corollary \ref{cor:oi-simulates-id} in Section~\ref{sec:oi-simulates-id} gives us a $t$-time $\OI$-algorithm $\Alg_\OI$ that computes a maximal fractional matching on the canonically ordered universal cover $(U_G,\preceq)$ for any loopy $\PO$-graph $G$ of maximum degree $\Delta$.
    \item[$\PO\leadsto\OI$:] Simulation (\ref{eq:po-simulates-oi}) in Section~\ref{sec:po-simulates-oi} queries the output of $\Alg_\OI$ only on $(U_G,\preceq)$. This gives us a $t$-time $\PO$-algorithm $\Alg_\PO$ that computes a maximal fractional matching on any loopy $\PO$-graph $G$ of maximum degree $\Delta$.
    \item[$\EC\leadsto\PO$:] The simple simulation in Section~\ref{sec:ec-simulates-po} gives us a $t$-time $\EC$-algorithm $\Alg_\EC$ that computes a maximal fractional matching on any loopy $\EC$-graph $G$ of maximum degree $\Delta/2$.
\end{itemize}
But now we can use the construction of Section~\ref{sec:lb-in-ec}: there is a loopy $\EC$-graph of maximum degree $\Delta/2$ where $\Alg_\EC$ runs for $\Omega(\Delta)$ rounds. Hence the running time of $\Alg$ is also $\Omega(\Delta)$.

\section*{Acknowledgements}

This work is supported in part by the Academy of Finland, Grants 132380 and 252018, and by the Research Funds of the University of Helsinki. The combinatorial proof in Appendix~\ref{app:tree-order} is joint work with Christoph Lenzen and Roger Wattenhofer.

\DeclareUrlCommand{\Doi}{\urlstyle{same}}
\renewcommand{\doi}[1]{\href{http://dx.doi.org/#1}{\footnotesize\sf doi:\Doi{#1}}}

\pagebreak
\appendix
\section{Proof of Lemma~\ref{lem:tree-order}} \label{app:tree-order}

We give two proofs for Lemma~\ref{lem:tree-order}, the second one of which we have not seen in print.

\subsection{Algebraic proof}

The tree $T$ can be thought of as a Cayley graph of the free group on $d$ generators, and the free group admits a linear order that is invariant under the group acting on itself by multiplication; for details, see Neumann~\cite{neumann49ordered} and the discussion in~\cite[\defaultS5]{goos12local-approximation}.

\subsection{Combinatorial proof}
\newcommand{\myarrow}{\!\rightsquigarrow\!}
\newcommand{\Vin}{V_{\textsf{in}}}

In $T$ there is a unique simple directed path $x\myarrow y$ between any two nodes $x,y\in V(T)$. We use $V(x\myarrow y)$ and $E(x\myarrow y)$ to denote the nodes and edges of the path. Also, we set $\Vin(x\myarrow y) := V(x\myarrow y)\smallsetminus \{x,y\}$. We will assign to each path $x\myarrow y$ an integer value, denoted $\llbracket x\myarrow y\rrbracket$, which will determine the relative order of the endpoints.

By definition, in the $\PO$ model, we are given the following linear orders:
\begin{itemize}[label=$-$,noitemsep]
    \item Each node $v \in V(T)$ has a linear order $\prec_v$ on its incident edges.
    \item Each edge $e \in E(T)$ has a linear order $\prec_e$ on its incident nodes.
\end{itemize}
For notational convenience, we extend these relations a little: for $v \in \Vin(x\myarrow y)$ we define $x\prec_v y \iff e\prec_v e'$, where $e$ is the last edge on the path $x\myarrow v$ and $e'$ is the first edge on the path $v\myarrow y$; similarly, for $e\in E(x\myarrow y)$, we define $x\prec_e y\iff x'\prec_e y'$, where $e=\{x',y'\}$ and $x'$ and $y'$ appear on the path $x\myarrow y$ in this order.

\begin{figure}
    \centering
    \includegraphics[page=15]{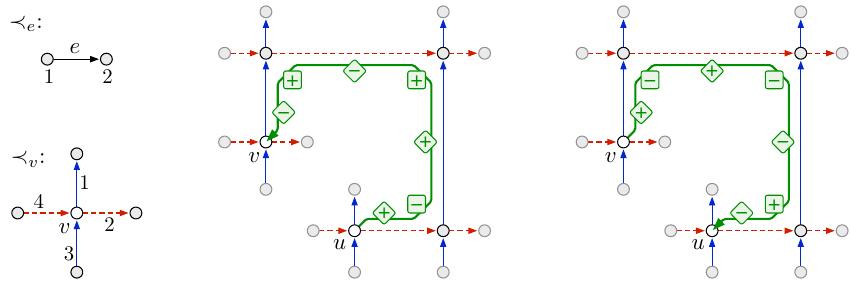}
    \caption{In this example, $\llbracket u\myarrow v \rrbracket = +1$, $\llbracket v\myarrow u \rrbracket = -1$, and hence $u \prec v$.}\label{fig:tree-order}
\end{figure}

For any statement $P$, we will use the following type of Iverson bracket notation:
\[
[P] := \begin{cases}
        +1 & \text{if $P$ is true}, \\
        -1 & \text{if $P$ is false}.
    \end{cases}
\]
We can now define
\begin{equation} \label{eq:path-value}
    \llbracket x\myarrow y\rrbracket\ := \sum_{e\in E(x\rightsquigarrow y)} [x\prec_e y]\quad +\ 
    \sum_{v\in \Vin(x\rightsquigarrow y)} [x\prec_v y].
\end{equation}
In particular, $\llbracket x\myarrow x\rrbracket = 0$. The linear order $\prec$ on $V(T)$ is now defined by setting
\[
    x \prec y \iff \llbracket x\myarrow y \rrbracket > 0.
\]
See Figure~\ref{fig:tree-order}. Next, we show that this is indeed a linear order.

\paragraph{Antisymmetry and totality.}

Since $[x\prec_v y] = -[y\prec_v x]$ and $[x\prec_e y] = -[y\prec_e x]$, we have the property that
\[
    \llbracket x\myarrow y\rrbracket = -\llbracket y\myarrow x\rrbracket.
\]
Moreover, if $x\neq y$, the first sum in \eqref{eq:path-value} is odd iff the second sum in \eqref{eq:path-value} is even. Therefore $\llbracket x\myarrow y\rrbracket$ is always odd; in particular, it is non-zero. These properties establish that either $x \prec y$ or $x \prec y$ (but never both).

\paragraph{Transitivity.}
Let $x,y,z \in V(T)$ be three distinct nodes with $x \prec y$ and $y \prec z$; we need to show that $x\prec z$. Denote by $v \in V(T)$ the unique node in the intersection of the paths $x\myarrow z$, $z\myarrow y$, and $y\myarrow x$.

Viewing the path $x\myarrow z$ piecewise as $x\myarrow v \myarrow z$ we write
\[
    \llbracket x\myarrow z\rrbracket
    = \llbracket x\myarrow v\rrbracket
    + [x\prec_v z]
    + \llbracket v\myarrow z\rrbracket,
\]
where it is understood that $[x\prec_v z] := 0$ in the degenerate cases where $v\in\{x,z\}$. Similar decompositions can be written for $z\myarrow y$ and $y\myarrow x$. Indeed, it is easily checked that
\[
      \llbracket x\myarrow z\rrbracket
    + \llbracket z\myarrow y\rrbracket
    + \llbracket y\myarrow x\rrbracket
    = [x\prec_v z]
    + [z\prec_v y]
    + [y\prec_v x].
\]
By assumption, $\llbracket z\myarrow y\rrbracket, \llbracket y\myarrow x\rrbracket \leq -1$, so we get
\[
\llbracket x\myarrow z\rrbracket
\geq 2 + [x\prec_v z]
    + [z\prec_v y]
    + [y\prec_v x].
\]
The only way the right hand side can be negative is if \[[x\prec_v z] = [z\prec_v y]= [y\prec_v x] = -1,\] but this is equivalent to having $z\prec_v x \prec_v y \prec_v z$, which is impossible. Hence $\llbracket x\myarrow z\rrbracket\geq 0$. But since $x\neq z$ we must have in fact that $x \prec z$.

\section{Derandomising Local Algorithms} \label{app:randomness}

As discussed in Section~\ref{sec:tricky-ids}, unbounded outputs require special care. In this Appendix we note that even though Naor and Stockmeyer~\cite{naor95what} assume bounded outputs, their result on derandomising local algorithms applies in our setting, too.

Let $\Alg$ be a randomised $t(\Delta)$-time algorithm that computes a maximal \FM on graphs of maximum degree $\Delta$ or possibly fails with some small probability. Given an assignment of random bit strings $\rho\colon V(G)\to\{0,1\}^*$ to the nodes of a graph $G$, denote by $\Alg^\rho$ the \emph{deterministic} algorithm that computes as $\Alg$, but uses $\rho$ for randomness.

The proof of Theorem 5.1 in~\cite{naor95what} is using the following fact whose proof we reproduce here for convenience.
\begin{lemma}[Naor and Stockmeyer]
For every $n$, there is an $n$-set $S_n\subseteq \N$ of identifiers and an assignment $\rho_n\colon S_n\to \{0,1\}^*$ such that $\Alg^{\rho_n}$ is correct on all graphs that have identifiers from $S_n$.
\end{lemma}
\begin{proof}
Denote by $k=k(n)$ the number of graphs $G$ with $V(G)\subseteq\{1,\ldots,n\}$. Let $X_1,\ldots,X_q\subseteq \N$ be any $q$ disjoint sets of size $n$. Suppose for the sake of contradiction that the claim is false for each $X_i$. That is, for any assignment $\rho\colon X_i\to\{0,1\}^*$ of random bits, $\Alg^\rho$ fails on at least one of the $k$ many graphs $G$ with $V(G)\subseteq X_i$. By averaging, this implies that for each $i$ there is a particular graph $G_i$, $V(G_i)\subseteq X_i$, on which $\Alg$ fails with probability at least $1/k$. Consider the graph $G$ that is the disjoint union of the graphs $G_1,\ldots,G_q$. Since $\Alg$ fails independently on each of the components $G_i$, the failure probability on $G$ is at least $1-(1-1/k)^q$. But this probability can be made arbitrarily close to $1$ by choosing a large enough $q$, which contradicts the correctness of $\Alg$.
\end{proof}

The deterministic algorithms $\Alg^{\rho_n}$ allow us to again obtain a $t(\Delta)$-time $\OI$-algorithm, which establishes the $\Omega(\Delta)$ lower bound for $\Alg$. Only small modifications to Section~\ref{sec:oi-simulates-id} are needed:
\begin{itemize}[label=$-$]
\item {\bf Step (i).}
Instead of the infinite set $I\subseteq\N$ as previously provided by Lemma~\ref{lem:ramsey}, we can use the finite Ramsey's theorem to find arbitrarily large sets $I_n\subseteq S_n$ (i.e., $|I_n|\to\infty$ as $n\to\infty$) with the property that $\Alg^{\rho_n}$ fully saturates the nodes of a loopy $\OI$-neighbourhood that has identifiers from $I_n$ (Lemma~\ref{lem:oi-saturate}).
\item {\bf Step (ii).}
Then, passing again to sufficiently sparse subsets $J_n\subseteq I_n$, we can reprove Lemma~\ref{lem:loopy-neigh} and Corollary~\ref{cor:oi-simulates-id}, which only require that $J$ is large enough.
\end{itemize}
This concludes the lower bound proof for randomised $\LOCAL$ algorithms.

\end{document}